\newif\ifBASIC
\BASICfalse
\newif\ifWP
\WPfalse
\newif\ifJOURNAL
\JOURNALfalse

\newif\ifFULL
\FULLfalse
\newif\ifLATIN
\LATINfalse

\BASICtrue	

\LATINtrue	

\newif\ifnotJOURNAL	
\notJOURNALtrue
\ifJOURNAL\notJOURNALfalse\fi

\newif\ifnotFULL	
\notFULLtrue
\ifFULL\notFULLfalse\fi

\newif\ifnotLATIN	
\notLATINtrue
\ifLATIN\notLATINfalse\fi

\ifBASIC
  \newcommand{\GTPi}{vovk/shafer:2008CAPM}
  \newcommand{\GTPii}{GTP2}
  \newcommand{\CTiv}{vovk:FS}
  \newcommand{\CTvii}{vovk:arXiv1109}
\fi
\ifWP
  \newcommand{\GTPi}{GTP1}
  \newcommand{\GTPii}{GTP2}
  \newcommand{\CTiv}{GTP28arXiv}
  \newcommand{\CTvii}{vovk:arXiv1109}
\fi
\ifJOURNAL
  \newcommand{\GTPi}{vovk/shafer:2008CAPM}
  \newcommand{\GTPii}{GTP2}
  \newcommand{\CTiv}{vovk:FS-short}
  \newcommand{\CTvii}{vovk:arXiv1109}
\fi

\ifnotLATIN
  
\fi
\ifLATIN
  
\fi

\newcommand{\acknowledge}{This research has been supported in part
by NWO Rubicon grant 680-50-1010.}

\ifBASIC
  \documentclass[10pt]{article}
  \usepackage{amsmath,amsfonts,amssymb,amsthm,latexsym,graphicx,epigraph}
\fi

\ifWP
  \documentclass[11pt]{article}
  \usepackage{amsmath,amsfonts,amssymb,amsthm,latexsym,graphicx,epigraph}
  \input{/Doc/Computing/Latex/gt.txt}
\fi

\ifJOURNAL
  \documentclass[dvips,aos]{imsart}
  \usepackage{amsmath,amsfonts,amssymb,amsthm,latexsym,graphicx,epigraph}
\fi

\ifFULL
  \usepackage{color}

  \newcommand{\bluebegin}{\begingroup\color{blue}}
  \newcommand{\blueend}{\endgroup}

\fi

\allowdisplaybreaks[4]

\newcommand{\Vladimir}{Vladimir}

\ifnotLATIN
  \input{OT2enc.def}
  
  \usepackage{CJK}
\fi

\newcommand{\dd}{\mathrm{d}}	

\newcommand{\K}{\mathcal{K}}	

\DeclareMathOperator{\III}{\boldsymbol{1}}  
\DeclareMathOperator{\e}{\mathrm{e}}	    

\newcommand{\bbbr}{\mathbb{R}}			

\newcommand{\bbbe}{\mathbb{E}}		
\DeclareMathOperator{\Expect}{\bbbe}

\newcommand{\bbbp}{\mathbb{P}}		
\DeclareMathOperator{\Prob}{\bbbp}

\DeclareMathOperator{\cov}{\textrm{cov}}		

\theoremstyle{plain}
\newtheorem{theorem}{Theorem}[section]
\newtheorem{proposition}[theorem]{Proposition}
\newtheorem{corollary}[theorem]{Corollary}
\newtheorem{lemma}[theorem]{Lemma}

\theoremstyle{definition}

\makeatletter
  \@addtoreset{equation}{section}
  
\makeatother

\ifBASIC
  \title{The Capital Asset Pricing Model as a corollary of the Black--Scholes model}
  \author{Vladimir Vovk}
  \date{September 23, 2011}
\fi

\ifWP
  \title{The Capital Asset Pricing Model as a corollary of the Black--Scholes model}
  \author{Vladimir Vovk}
  
\fi

\begin{document}
\ifJOURNAL
  \begin{frontmatter}
    \title{The Capital Asset Pricing Model as a corollary of the Black--Scholes model}
    \runtitle{Capital Asset Pricing Model}
    \author{\fnms{Vladimir} \snm{Vovk}\ead[label=e1]{vovk@cs.rhul.ac.uk}\thanksref{t1}}
    \thankstext{t1}{\acknowledge}
    \address{Department of Computer Science\\Royal Holloway, University of London\\
      Egham, Surrey TW20 0EX\\United Kingdom\\\printead{e1}}
    \affiliation{Royal Holloway, University of London}
    \runauthor{Vladimir Vovk}
\fi

\ifnotJOURNAL
  \maketitle
\fi

\begin{abstract}
  We consider a financial market in which two securities are traded:
  a stock and an index.
  Their prices are assumed to satisfy the Black--Scholes model.
  Besides assuming that the index is a tradable security,
  we also assume that it is \emph{efficient}, in the following sense:
  we do not expect a prespecified self-financing trading strategy
  whose wealth is almost surely nonnegative at all times
  to outperform the index greatly.
  We show that, for a long investment horizon, the appreciation rate of the stock
  has to be close to the interest rate (assumed constant)
  plus the covariance between the volatility vectors of the stock and the index.
  This contains both a version of the Capital Asset Pricing Model
  and our earlier result that the equity premium
  is close to the squared volatility of the index.
\end{abstract}

\ifJOURNAL

  \end{frontmatter}
\fi

\epigraph{For me, the strongest evidence
suggesting that markets are generally quite efficient
is that professional investors do not beat the market.}
{Burton G. Malkiel \cite{malkiel:2005}}

\section{Introduction}

This article continues study of the \emph{efficient index hypothesis} (\emph{EIH}),
introduced in \cite{shafer/vovk:2001} (under a different name)
and later studied in \cite{\GTPi} and \cite{\CTvii}.
The EIH is a hypothesis about a specific index $I_t$, such as FTSE 100.
Let $\Sigma$ be any trading strategy that is \emph{prudent},
in the sense of its wealth process being nonnegative almost surely
at all times.
(We consider only self-financing trading strategies in this article.)
Trading occurs over the time period $[0,T]$,
where the investment horizon $T>0$ is fixed throughout the article,
and we assume that $I_0>0$.
The EIH says that, as long as $\Sigma$ is chosen in advance
and its initial wealth $\K_0$ is positive, $\K_0>0$,
we do not expect $\K_T/\K_0$,
where $\K_T$ is its final wealth,
to be much larger than $I_T/I_0$.

The EIH is similar to the Efficient Market Hypothesis
(EMH; see \cite{fama:1970} and \cite{malkiel:2005} for surveys)
and in some form is considered to be evidence in favour of the EMH
(see the epigraph above).
But it is also an interesting hypothesis in its own right.
For example, in this article we will see
that in the framework of the Black--Scholes model
it implies a version of the Capital Asset Pricing Model (CAPM),
whereas the EMH is almost impossible to disentangle
from the CAPM or similar asset pricing models
(see, e.g., \cite{fama:1970}, III.A.6).

Several remarks about the EIH are in order
(following \cite{\CTvii}):
\begin{itemize}
\item
  Our mathematical results do not depend on the EIH,
  which is only used in their interpretation.
  They are always of the form: either some interesting relation holds
  or a given prudent trading strategy outperforms the index greatly
  (almost surely or with a high probability).
\item
  Even when using the EIH in the interpretation of our results,
  we do not need the full EIH:
  we apply it only to very basic trading strategies.
\item
  Our prudent trading strategies can still lose all their initial wealth
  (they are only prudent in the sense of not losing more than the initial wealth).
  A really prudent investor would invest only part of her capital
  in such strategies.
\end{itemize}

We start the rest of the article by proving a result
about the ``theoretical performance deficit''
(in the terminology of \cite{\GTPi})
of a stock $S_t$ as compared with the index $I_t$,
Namely, in Section~\ref{sec:TPD} we show that,
for a long investment horizon and assuming the EIH,
\begin{equation}\label{eq:TPD}
  \ln\frac{S_T/S_0}{I_T/I_0}
  \approx
  -\frac{\left\|\sigma_S-\sigma_I\right\|^2}{2}
  T,
\end{equation}
where $I_0$ is assumed positive
and $\sigma_S$ and $\sigma_I$ are the volatility vectors
(formally defined in Section~\ref{sec:TPD})
for the stock and the index.
We can call $\left\|\sigma_S-\sigma_I\right\|^2/2$
the theoretical performance deficit
as it can be attributed to insufficient diversification of $S_t$
as compared to $I_t$.
Section~\ref{sec:CAPM} deduces a version of the CAPM
from (\ref{eq:TPD});
this version is similar to the one obtained in \cite{\GTPi}
but our interpretation and methods are very different.
Section~\ref{sec:conclusion} concludes.

\section{Theoretical performance deficit}
\label{sec:TPD}

The value of the index at time $t$ is denoted $I_t$
and the value of the stock is denoted $S_t$.
We assume that these two securities satisfy
the multi-dimensional Black--Scholes model
\begin{equation}\label{eq:physical-1}
  \begin{cases}
    \frac{\dd I_t}{I_t}
    =
    \mu_I\dd t + \sigma_{I,1}\dd W^1_t + \cdots + \sigma_{I,d}\dd W^d_t
    \\
    \frac{\dd S_t}{S_t}
    =
    \mu_S\dd t + \sigma_{S,1}\dd W^1_t + \cdots + \sigma_{S,d}\dd W^d_t,
  \end{cases}
\end{equation}
where $W^1,\ldots,W^d$ are independent standard Brownian motions.
For simplicity, we also assume,
without loss of generality, that $I_0=1$ and $S_0=1$.
The parameters of the model are the appreciation rates $\mu_I,\mu_S\in\bbbr$
and the volatility vectors $\sigma_I:=(\sigma_{I,1},\ldots,\sigma_{I,d})^{\mathrm{T}}$
and $\sigma_S:=(\sigma_{S,1},\ldots,\sigma_{S,d})^{\mathrm{T}}$.
We assume $\sigma_I\ne\sigma_S$, $\sigma_I\ne0$, and $\sigma_S\ne0$.
The number of ``sources of randomness'' $W^1,\ldots,W^d$ in our market is $d\ge2$.
The interest rate $r$ is constant.
We interpret $\e^{rt}$ as the price of a zero-coupon bond at time $t$.

Let us say that a prudent trading strategy \emph{beats the index by a factor of $c$}
if its wealth process $\K_t$ satisfies $\K_0>0$ and $\K_T/\K_0=c I_T$.
Let $N_{0,1}$ be the standard Gaussian distribution on $\bbbr$
and $z_p$, $p>0$, be its upper $p$-quantile,
defined by the requirement $\Prob(\xi\ge z_p)=p$, $\xi\sim N_{0,1}$, when $p\in(0,1)$,
and defined as $-\infty$ when $p\ge1$.
We start from the following proposition.
\begin{proposition}\label{prop:2-sided}
  Let $\delta>0$.
  There is a prudent trading strategy
  $\Sigma=\Sigma(\sigma_I,\sigma_S,r,T,\delta)$
  that, almost surely, beats the index by a factor of $1/\delta$ unless
  \begin{equation}\label{eq:2-sided}
    \left|
      \ln\frac{S_T}{I_T}
      +
      \frac{\left\|\sigma_S-\sigma_I\right\|^2}{2}
      T
    \right|
    <
    z_{\delta/2}
    \left\|\sigma_S-\sigma_I\right\|
    \sqrt{T}.
  \end{equation}
\end{proposition}

We assumed $\sigma_S\ne0$,
but Proposition~\ref{prop:2-sided} remains true
when applied to the bond $B_t:=\e^{rt}$ in place of the stock $S_t$.
In this case (\ref{eq:2-sided}) reduces to
\begin{equation}\label{eq:recover}
  \left|
    \ln\frac{I_T}{\e^{rT}}
    -
    \frac{\left\|\sigma_I\right\|^2}{2}
    T
  \right|
  <
  z_{\delta/2}
  \left\|\sigma_I\right\|
  \sqrt{T}.
\end{equation}
Informally, (\ref{eq:recover}) says that the index outperforms the bond
approximately by a factor of $\e^{\left\|\sigma_I\right\|^2T/2}$.
For a proof of this statement
(which is similar to, but simpler than, the proof of Proposition~\ref{prop:2-sided}
given later in this section),
see \cite{\CTvii}, Proposition~2.1.

In the next section we will need the following one-sided version
of Proposition~\ref{prop:2-sided}.
\begin{proposition}\label{prop:1-sided}
  Let $\delta>0$.
  There is a prudent trading strategy
  $\Sigma=\Sigma(\sigma_I,\sigma_S,r,T,\delta)$
  that, almost surely, beats the index by a factor of $1/\delta$ unless
  \begin{equation}\label{eq:1-sided}
    \ln\frac{S_T}{I_T}
    +
    \frac{\left\|\sigma_S-\sigma_I\right\|^2}{2}
    T
    <
    z_{\delta}
    \left\|\sigma_S-\sigma_I\right\|
    \sqrt{T}.
  \end{equation}
  There is another prudent trading strategy
  $\Sigma=\Sigma(\sigma_I,\sigma_S,r,T,\delta)$
  that, almost surely, beats the index by a factor of $1/\delta$ unless
  \begin{equation*} 
    \ln\frac{S_T}{I_T}
    +
    \frac{\left\|\sigma_S-\sigma_I\right\|^2}{2}
    T
    >
    -z_{\delta}
    \left\|\sigma_S-\sigma_I\right\|
    \sqrt{T}.
  \end{equation*}
\end{proposition}

In the rest of this section we will prove Proposition~\ref{prop:2-sided}
(Proposition~\ref{prop:1-sided} can be proved analogously).
Without loss of generality suppose $\delta\in(0,1)$.
We let $W_t$ stand for the $d$-dimensional Brownian motion
$W_t:=(W^1_t,\ldots,W^d_t)^{\mathrm{T}}$.
The market (\ref{eq:physical-1}) is incomplete when $d>2$,
as it has too many sources of randomness,
so we start from removing superfluous sources of randomness.

The standard solution to (\ref{eq:physical-1}) is
\begin{equation}\label{eq:solution-1}
  \begin{cases}
    I_t
    =
    \e^{(\mu_I-\left\|\sigma_I\right\|^2/2)t+\sigma_I\cdot W_t}\\
    S_t
    =
    \e^{(\mu_S-\left\|\sigma_S\right\|^2/2)t+\sigma_S\cdot W_t}.
  \end{cases}
\end{equation}
Choose two vectors $e^1,e^2\in\bbbr^d$ that form an orthonormal basis
in the 2-dimensional subspace of $\bbbr^d$ spanned by $\sigma_I$ and $\sigma_S$.
Set $\bar W^1_t:=e^1\cdot W_t$ and $\bar W^2_t:=e^2\cdot W_t$;
these are standard independent Brownian motions.
Let the decompositions of $\sigma_I$ and $\sigma_S$ in the basis $(e^1,e^2)$ be
$\sigma_I=\bar\sigma_{I,1}e^1+\bar\sigma_{I,2}e^2$
and $\sigma_S=\bar\sigma_{S,1}e^1+\bar\sigma_{S,2}e^2$.
Define $\bar\sigma_I:=(\bar\sigma_{I,1},\bar\sigma_{I,2})^{\mathrm{T}}\in\bbbr^2$
and $\bar\sigma_S:=(\bar\sigma_{S,1},\bar\sigma_{S,2})^{\mathrm{T}}\in\bbbr^2$,
and define $\bar W_t$ as the 2-dimensional Brownian motion
$\bar W_t:=(\bar W^1_t,\bar W^2_t)^{\mathrm{T}}$.
We can now rewrite (\ref{eq:solution-1}) as
\begin{equation*} 
  \begin{cases}
    I_t
    =
    \e^{(\mu_I-\left\|\bar\sigma_I\right\|^2/2)t+\bar\sigma_I\cdot\bar W_t}\\
    S_t
    =
    \e^{(\mu_S-\left\|\bar\sigma_S\right\|^2/2)t+\bar\sigma_S\cdot\bar W_t}.
  \end{cases}
\end{equation*}
In terms of our new parameters and Brownian motions,
(\ref{eq:physical-1}) can be rewritten as
\begin{equation}\label{eq:physical-2}
  \begin{cases}
    \frac{\dd I_t}{I_t}
    =
    \mu_I\dd t + \bar\sigma_{I,1}\dd\bar W^1_t + \bar\sigma_{I,2}\dd\bar W^2_t
    \\
    \frac{\dd S_t}{S_t}
    =
    \mu_S\dd t + \bar\sigma_{S,1}\dd\bar W^1_t + \bar\sigma_{S,2}\dd\bar W^2_t.
  \end{cases}
\end{equation}
The risk-neutral version of (\ref{eq:physical-2}) is
\begin{equation*} 
  \begin{cases}
    \frac{\dd I_t}{I_t}
    =
    r\dd t + \bar\sigma_{I,1}\dd\bar W^1_t + \bar\sigma_{I,2}\dd\bar W^2_t
    \\
    \frac{\dd S_t}{S_t}
    =
    r\dd t + \bar\sigma_{S,1}\dd\bar W^1_t + \bar\sigma_{S,2}\dd\bar W^2_t,
  \end{cases}
\end{equation*}
whose solution is
$$
  \begin{cases}
    I_t
    =
    \e^{(r-\left\|\bar\sigma_I\right\|^2/2)t+\bar\sigma_I\cdot\bar W_t}\\
    S_t
    =
    \e^{(r-\left\|\bar\sigma_S\right\|^2/2)t+\bar\sigma_S\cdot\bar W_t}.
  \end{cases}
$$

Let $b\in\bbbr$
and let $\III\{\ldots\}$ be defined to be $1$
if the condition in the curly braces is satisfied
and $0$ otherwise.
The Black--Scholes price at time $0$ of the European contingent claim
paying $I_T\III\{S_T/I_T\ge b\}$ at time $T$ is
\begin{multline}\label{eq:to-continue}
  \e^{-rT}
  \Expect
  \left(
    \e^{(r-\left\|\bar\sigma_I\right\|^2/2)T+\sqrt{T}\bar\sigma_I\cdot\xi}
    \III
    \left\{
      \frac
        {\e^{(r-\left\|\bar\sigma_S\right\|^2/2)T+\sqrt{T}\bar\sigma_S\cdot\xi}}
        {\e^{(r-\left\|\bar\sigma_I\right\|^2/2)T+\sqrt{T}\bar\sigma_I\cdot\xi}}
	\ge
	b
    \right\}
  \right)\\
  =
  \e^{-\left\|\bar\sigma_I\right\|^2T/2}
  \Expect
  \left(
    \e^{\sqrt{T}\bar\sigma_I\cdot\xi}
    \III
    \left\{
      \sqrt{T}(\bar\sigma_S-\bar\sigma_I)\cdot\xi
      \ge
      \ln b
      +
      \frac{\left\|\bar\sigma_S\right\|^2-\left\|\bar\sigma_I\right\|^2}{2}
      T
    \right\}
  \right),
\end{multline}
where $\xi\sim N^2_{0,1}$.
To continue our calculations,
we will need the following lemma.
\begin{lemma}
  Let $u,v\in\bbbr^2$, $v\ne0$, $c\in\bbbr$, and $\xi\sim N_{0,1}^2$.
  Then
  $$
    \Expect
    \left(
      \e^{u\cdot\xi}
      \III\{v\cdot\xi\ge c\}
    \right)
    =
    \e^{\left\|u\right\|^2/2}
    F
    \left(
      \frac{u\cdot v-c}{\left\|v\right\|}
    \right),
  $$
  where $F$ is the distribution function of $N_{0,1}$.
\end{lemma}
\begin{proof}
  This follows from
  \begin{align*}
    \Expect
    \left(
      \e^{u\cdot\xi}
      \III\{v\cdot\xi\ge c\}
    \right)
    &=
    \frac{1}{2\pi}
    \int_{\bbbr^2}
    \e^{u\cdot z}
    \III\{v\cdot z\ge c\}
    \e^{-\left\|z\right\|^2/2}
    \dd z\\
    &=
    \frac{1}{2\pi}
    \e^{\left\|u\right\|^2/2}
    \int_{\bbbr^2}
    \III\{v\cdot z\ge c\}
    \e^{-\left\|z-u\right\|^2/2}
    \dd z\\
    &=
    \frac{1}{2\pi}
    \e^{\left\|u\right\|^2/2}
    \int_{\bbbr^2}
    \III\left\{v\cdot w\ge c-u\cdot v\right\}
    \e^{-\left\|w\right\|^2/2}
    \dd w\\
    &=
    \e^{\left\|u\right\|^2/2}
    \Prob
    \left(
      \frac{v}{\left\|v\right\|}\cdot\xi
      \ge
      \frac{c-u\cdot v}{\left\|v\right\|}
    \right)\\
    &=
    \e^{\left\|u\right\|^2/2}
    F
    \left(
      \frac{u\cdot v-c}{\left\|v\right\|}
    \right).
    \qedhere
  \end{align*}
\end{proof}

Now we can rewrite (\ref{eq:to-continue}) as
$$
  F
  \left(
    \frac
    {
      T\bar\sigma_I\cdot(\bar\sigma_S-\bar\sigma_I)
      -
      \ln b
      -
      \frac{\left\|\bar\sigma_S\right\|^2-\left\|\bar\sigma_I\right\|^2}{2}
      T
    }
    {
      \left\|\sqrt{T}(\bar\sigma_S-\bar\sigma_I)\right\|
    }
  \right)
  =
  F
  \left(
    -\frac
    {
      \frac{\left\|\bar\sigma_S-\bar\sigma_I\right\|^2}{2} T
      +
      \ln b
    }
    {
      \left\|\bar\sigma_S-\bar\sigma_I\right\| \sqrt{T}
    }
  \right).
$$
Let us define $b$ by the requirement
$$
  \frac
  {
    \frac{\left\|\bar\sigma_S-\bar\sigma_I\right\|^2}{2}T
    +
    \ln b
  }
  {
    \left\|\bar\sigma_S-\bar\sigma_I\right\| \sqrt{T}
  }
  =
  z_{\delta/2},
$$
i.e.,
\begin{equation}\label{eq:b}
  \ln b
  =
  -\frac{\left\|\bar\sigma_S-\bar\sigma_I\right\|^2}{2}
  T
  +
  z_{\delta/2}
  \left\|\bar\sigma_S-\bar\sigma_I\right\|
  \sqrt{T}.
\end{equation}
As the Black--Scholes price of the European contingent claim $I_T\III\{S_T/I_T\ge b\}$
is $\delta/2$,
there is a prudent trading strategy $\Sigma_1$ with initial wealth $\delta/2$
that almost surely beats the index by a factor of $2/\delta$
if $S_T/I_T\ge b$.

Now let $a\in\bbbr$ and consider the European contingent claim
paying $I_T\III\{S_T/I_T\le a\}$.
Replacing ``${}\ge b$'' by ``${}\le a$'' and ``${}\ge\ln b$'' by ``${}\le\ln a$''
in (\ref{eq:to-continue})
and defining $a$ to satisfy
\begin{equation*} 
  \ln a
  =
  -\frac{\left\|\bar\sigma_S-\bar\sigma_I\right\|^2}{2}
  T
  -
  z_{\delta/2}
  \left\|\bar\sigma_S-\bar\sigma_I\right\|
  \sqrt{T}
\end{equation*}
in place of (\ref{eq:b}),
we obtain a prudent trading strategy $\Sigma_2$ that starts from $\delta/2$
and almost surely beats the index by a factor of $2/\delta$
if $S_T/I_T\le a$.
The sum $\Sigma:=\Sigma_1+\Sigma_2$ will beat the index by a factor of $1/\delta$
if $S_T/I_T\notin(a,b)$.
This completes the proof of Proposition~\ref{prop:2-sided}.

\section{Capital Asset Pricing Model}
\label{sec:CAPM}

In this section we will derive a version of the CAPM
from the results of the previous section.
Our argument will be similar to that of Section 3 of \cite{\CTvii}.
\begin{proposition}\label{prop:mu}
  For each $\delta>0$
  there exists a prudent trading strategy
  $\Sigma=\Sigma(\sigma_I,\sigma_S,r,T,\delta)$
  that satisfies the following condition.
  For each $\epsilon>0$, either
  \begin{equation}\label{eq:mu}
    \left|
      \mu_S - \mu_I + \left\|\sigma_I\right\|^2 - \sigma_S\cdot\sigma_I
    \right|
    <
    \frac{(z_{\delta/2}+z_{\epsilon}) \left\|\sigma_S-\sigma_I\right\|}{\sqrt{T}}
  \end{equation}
  or $\Sigma$ beats the index by a factor of at least $1/\delta$
  with probability at least $1-\epsilon$.
\end{proposition}
\begin{proof}
  Suppose (\ref{eq:mu}) is violated;
  we are required to prove that some prudent trading strategy
  beats the index by a factor of at least $1/\delta$
  with probability at least $1-\epsilon$.
  We have either
  \begin{equation}\label{eq:mu-1}
    \mu_S - \mu_I + \left\|\sigma_I\right\|^2 - \sigma_S\cdot\sigma_I
    \ge
    \frac{(z_{\delta/2}+z_{\epsilon}) \left\|\sigma_S-\sigma_I\right\|}{\sqrt{T}}
  \end{equation}
  or
  \begin{equation}\label{eq:mu-2}
    \mu_S - \mu_I + \left\|\sigma_I\right\|^2 - \sigma_S\cdot\sigma_I
    \le
    -\frac{(z_{\delta/2}+z_{\epsilon}) \left\|\sigma_S-\sigma_I\right\|}{\sqrt{T}}.
  \end{equation}
  The two cases are analogous, and we will assume, for concreteness,
  that (\ref{eq:mu-1}) holds.

  As (\ref{eq:solution-1}) solves (\ref{eq:physical-1}),
  we have
  \begin{equation}\label{eq:ln-ratio}
    \ln\frac{S_T}{I_T}
    =
    (\mu_S-\mu_I)T
    +
    \frac{\left\|\sigma_I\right\|^2-\left\|\sigma_S\right\|^2}{2}T
    +
    \sqrt{T}(\sigma_S-\sigma_I)\cdot\xi,
  \end{equation}
  where $\xi\sim N_{0,1}^d$.
  In combination with (\ref{eq:mu-1}) this gives
  \begin{align}
    \ln\frac{S_T}{I_T}
    &\ge
    \left(
      -\left\|\sigma_I\right\|^2 + \sigma_S\cdot\sigma_I
      +
      \frac{(z_{\delta/2}+z_{\epsilon}) \left\|\sigma_S-\sigma_I\right\|}{\sqrt{T}}
    \right)
    T\notag\\
    &\quad+
    \frac{\left\|\sigma_I\right\|^2-\left\|\sigma_S\right\|^2}{2}T
    +
    \sqrt{T}(\sigma_S-\sigma_I)\cdot\xi\notag\\
    &=
    -\frac{\left\|\sigma_S-\sigma_I\right\|^2}{2}T
    +
    (z_{\delta/2}+z_{\epsilon}) \left\|\sigma_S-\sigma_I\right\| \sqrt{T}
    +
    \sqrt{T}(\sigma_S-\sigma_I)\cdot\xi.
    \label{eq:intermediate-1}
  \end{align}

  Let $\Sigma$ be a prudent trading strategy that, almost surely,
  beats the index by a factor of $1/\delta$
  unless (\ref{eq:2-sided}) holds.
  It is sufficient to prove that the probability of (\ref{eq:2-sided})
  is at most $\epsilon$.
  In combination with (\ref{eq:intermediate-1}),
  (\ref{eq:2-sided}) implies
  \begin{equation}\label{eq:intermediate-2}
    z_{\delta/2} \left\|\sigma_S-\sigma_I\right\| \sqrt{T}
    >
    (z_{\delta/2}+z_{\epsilon}) \left\|\sigma_S-\sigma_I\right\| \sqrt{T}
    +
    \sqrt{T}(\sigma_S-\sigma_I)\cdot\xi,
  \end{equation}
  i.e.,
  \begin{equation}\label{eq:intermediate-3}
    \frac{\sigma_S-\sigma_I}{\left\|\sigma_S-\sigma_I\right\|}
    \cdot
    \xi
    <
    -z_{\epsilon}.
  \end{equation}
  The probability of the last event is $\epsilon$.
\end{proof}

Allowing the strategy $\Sigma$ to depend, additionally,
on $\mu_I$, $\mu_S$, and $\epsilon$,
we can improve (\ref{eq:mu}) replacing $\delta/2$ by $\delta$.
\begin{proposition}\label{prop:mu-bis}
  Let $\delta>0$ and $\epsilon>0$.
  Unless
  \begin{equation}\label{eq:mu-bis}
    \left|
      \mu_S - \mu_I + \left\|\sigma_I\right\|^2 - \sigma_S\cdot\sigma_I
    \right|
    <
    \frac{(z_{\delta}+z_{\epsilon}) \left\|\sigma_S-\sigma_I\right\|}{\sqrt{T}},
  \end{equation}
  there exists a prudent trading strategy
  $\Sigma=\Sigma(\mu_I,\mu_S,\sigma_I,\sigma_S,r,T,\delta,\epsilon)$
  that beats the index by a factor of at least $1/\delta$
  with probability at least $1-\epsilon$.
\end{proposition}
\begin{proof}
  We modify slightly the proof of Proposition~\ref{prop:mu}:
  as $\Sigma$ we now take a prudent trading strategy 
  that, almost surely, beats the index by a factor of $1/\delta$
  unless (\ref{eq:1-sided}) holds.
  Combining (\ref{eq:intermediate-1}) with $\delta$ in place of $\delta/2$
  and (\ref{eq:1-sided})
  we get (\ref{eq:intermediate-2}) with $\delta$ in place of $\delta/2$,
  and we still have (\ref{eq:intermediate-3}).
  Notice that $\Sigma$ now depends on which of the two cases,
  (\ref{eq:mu-1}) or (\ref{eq:mu-2})
  (with $\delta$ in place of $\delta/2$),
  holds.
\end{proof}

Propositions~\ref{prop:mu} and~\ref{prop:mu-bis} contain a version of the CAPM.
But before stating CAPM-type results formally
as corollaries of Proposition~\ref{prop:mu-bis}
(we do not state the analogous easy corollaries of Proposition~\ref{prop:mu}),
we will discuss them informally,
to give us a sense of direction.

Assuming $\delta\ll1$, $\epsilon\ll1$, and $T\gg1$,
we can interpret (\ref{eq:mu-bis}) as saying that
\begin{equation}\label{eq:1}
  \mu_S
  \approx
  \mu_I - \left\|\sigma_I\right\|^2 + \sigma_S\cdot\sigma_I.
\end{equation}
This approximate equality is applicable to the bond as well as the stock
(by results of \cite{\CTvii}),
which gives
\begin{equation}\label{eq:2}
  \mu_I
  \approx
  r + \left\|\sigma_I\right\|^2.
\end{equation}
Combining (\ref{eq:1}) and (\ref{eq:2}) we obtain
\begin{equation}\label{eq:3}
  \mu_S
  \approx
  r + \sigma_S\cdot\sigma_I.
\end{equation}
And combining (\ref{eq:3}) and (\ref{eq:2}) we obtain
\begin{equation}\label{eq:4}
  \mu_S
  \approx
  r
  +
  \frac{\sigma_S\cdot\sigma_I}{\left\|\sigma_I\right\|^2}
  (\mu_I-r).
\end{equation}

Equation (\ref{eq:4}) is a continuous-time version of the CAPM.
The standard Sharpe--Lintner CAPM (see, e.g., \cite{fama/french:2004}, pp.~28--29)
can be written in the form
\begin{equation}\label{eq:5}
  \Expect(R_S)
  =
  r
  +
  \frac{\cov(R_S,R_I)}{\sigma^2(R_I)}
  \left(
    \Expect(R_I) - r
  \right),
\end{equation}
where $R_S$ and $R_I$ are the returns of a risky asset and the market portfolio,
respectively.
The correspondence between (\ref{eq:4}) and (\ref{eq:5})
is obvious\ifnotFULL.\fi

Now we state formal counterparts of (\ref{eq:2})--(\ref{eq:4}).
The following proposition,
which would have been a corollary of Proposition~\ref{prop:mu-bis}
had we allowed $\sigma_S=0$,
is proved in \cite{\CTvii}, Proposition~3.2.
\begin{proposition}\label{prop:index}
  Let $\delta>0$ and $\epsilon>0$.
  Unless
  \begin{equation}\label{eq:index}
    \left|
      \mu_I - r - \left\|\sigma_I\right\|^2
    \right|
    <
    \frac{(z_{\delta}+z_{\epsilon}) \left\|\sigma_I\right\|}{\sqrt{T}},
  \end{equation}
  there exists a prudent trading strategy
  $\Sigma=\Sigma(\mu_I,\sigma_I,r,T,\delta,\epsilon)$
  that beats the index by a factor of at least $1/\delta$
  with probability at least $1-\epsilon$.
\end{proposition}

The following two corollaries of Proposition~\ref{prop:mu-bis}
assert existence of trading strategies
that depend on ``everything'',
namely, on $\mu_I$, $\mu_S$, $\sigma_I$, $\sigma_S$, $r$, $T$, $\delta$, and $\epsilon$.
The first corollary formalizes (\ref{eq:3}).
\begin{corollary}\label{cor:CAPM-1}
  Let $\delta>0$ and $\epsilon>0$.
  Unless
  \begin{equation}\label{eq:CAPM-1}
    \left|
      \mu_S - r - \sigma_S\cdot\sigma_I
    \right|
    <
    (z_{\delta}+z_{\epsilon})
    \frac{\left\|\sigma_I\right\|+\left\|\sigma_S-\sigma_I\right\|}{\sqrt{T}},
  \end{equation}
  there exists a prudent trading strategy
  that beats the index by a factor of at least $\frac{1}{2\delta}$
  with probability at least $1-\epsilon$.
\end{corollary}
\begin{proof}
  Let $\Sigma_1$ be a prudent trading strategy
  satisfying the condition of Proposition~\ref{prop:mu-bis},
  and let $\Sigma_2$ be a prudent trading strategy
  satisfying the condition of Proposition~\ref{prop:index}.
  Without loss of generality suppose that the initial wealth
  of both strategies is 1.
  Then $\Sigma_1+\Sigma_2$ will beat the index
  by a factor of at least $\frac{1}{2\delta}$
  with probability at least $1-\epsilon$
  unless both (\ref{eq:mu-bis}) and (\ref{eq:index}) hold.
  The conjunction of (\ref{eq:mu-bis}) and (\ref{eq:index})
  implies (\ref{eq:CAPM-1}).
\end{proof}

Finally, we have a corollary formalizing the CAPM (\ref{eq:4}).
\begin{corollary} 
  Let $\delta>0$ and $\epsilon>0$.
  Unless
  \begin{equation*} 
    \left|
      \mu_S - r - \frac{\sigma_S\cdot\sigma_I}{\left\|\sigma_I\right\|^2} (\mu_I-r)
    \right|
    \le
    (z_{\delta}+z_{\epsilon})
    \frac
    {
      \left\|\sigma_I\right\|
      +
      \left\|\sigma_S\right\|
      +
      \left\|\sigma_S-\sigma_I\right\|
    }
    {\sqrt{T}},
  \end{equation*}
  there exists a prudent trading strategy
  that beats the index by a factor of at least $\frac{1}{3\delta}$
  with probability at least $1-\epsilon$.
\end{corollary}
\begin{proof}
  Let $\Sigma_1$ be a prudent trading strategy
  satisfying the condition of Proposition~\ref{prop:index}
  and $\Sigma_2$ be a prudent trading strategy
  satisfying the condition of Corollary~\ref{cor:CAPM-1}.
  Without loss of generality suppose that the initial wealth of $\Sigma_1$ is 1
  and the initial wealth of $\Sigma_2$ is 2.
  Then $\Sigma_1+\Sigma_2$ will beat the index
  by a factor of at least $\frac{1}{3\delta}$
  with probability at least $1-\epsilon$
  unless both (\ref{eq:index}) and (\ref{eq:CAPM-1}) hold.
  The conjunction of (\ref{eq:index}) and (\ref{eq:CAPM-1}) implies
  \begin{align*}
    &\left|
      \mu_S - r - \frac{\sigma_S\cdot\sigma_I}{\left\|\sigma_I\right\|^2} (\mu_I-r)
    \right|\\
    &\le
    \left|
      \mu_S - r
      -
      \frac{\sigma_S\cdot\sigma_I}{\left\|\sigma_I\right\|^2}
      \left\|\sigma_I\right\|^2
    \right|
    +
    \frac{\left|\sigma_S\cdot\sigma_I\right|}{\left\|\sigma_I\right\|^2}
    \frac{(z_{\delta}+z_{\epsilon}) \left\|\sigma_I\right\|}{\sqrt{T}}\\
    &\le
    (z_{\delta}+z_{\epsilon})
    \frac{\left\|\sigma_I\right\|+\left\|\sigma_S-\sigma_I\right\|}{\sqrt{T}}
    +
    \frac{\left|\sigma_S\cdot\sigma_I\right|}{\left\|\sigma_I\right\|}
    \frac{(z_{\delta}+z_{\epsilon})}{\sqrt{T}}\\
    &\le
    (z_{\delta}+z_{\epsilon})
    \frac
    {
      \left\|\sigma_I\right\|
      +
      \left\|\sigma_S\right\|
      +
      \left\|\sigma_S-\sigma_I\right\|
    }
    {\sqrt{T}}.
    \qedhere
  \end{align*}
\end{proof}

\section{Conclusion}
\label{sec:conclusion}

Let us summarize our results at the informal level of approximate equalities
such as (\ref{eq:1})--(\ref{eq:4}).
At this level, our only two results are the CAPM (\ref{eq:4})
and the equity premium relation (\ref{eq:2}) (established earlier in \cite{\CTvii});
the rest follows.
Indeed, (\ref{eq:4}) and (\ref{eq:2}) imply (\ref{eq:3}),
and (\ref{eq:3}) and (\ref{eq:2}) imply (\ref{eq:1}).
The crude form (\ref{eq:TPD}) of (\ref{eq:2-sided})
also follows from (\ref{eq:4}) and (\ref{eq:2}):
just combine the crude form
\begin{equation*}
  \ln\frac{S_T}{I_T}
  \approx
  (\mu_S-\mu_I)T
  +
  \frac{\left\|\sigma_I\right\|^2-\left\|\sigma_S\right\|^2}{2}T
\end{equation*}
of (\ref{eq:ln-ratio}) with (\ref{eq:1}).

An alternative, simpler, summary of our results at the informal level
is given by the approximate equality (\ref{eq:3})
in which we allow $S=I$.
We can allow $S=I$ even in Corollary~\ref{cor:CAPM-1}:
when $S=I$, it reduces to Proposition~\ref{prop:index}.
The approximate equality (\ref{eq:3}) implies both (\ref{eq:2})
(it is a special case for $S:=I$)
and (\ref{eq:4})
(combine (\ref{eq:3}) and (\ref{eq:2})).
Therefore, at the informal level
Corollary~\ref{cor:CAPM-1} is the core result of this article.

\section*{Acknowledgments}

\acknowledge

\end{document}